\documentclass[a4paper,UKenglish]{lipics-v2016}

\usepackage{anyfontsize}
\usepackage{microtype}
\usepackage{mathtools}
\usepackage{amsfonts}
\usepackage[usenames,dvipsnames]{xcolor}
\usepackage[toc,page]{appendix}
\usepackage{tikz}
\usepackage{esvect}
\usepackage{xcolor}
\usepackage{hyperref}
\usepackage[ruled]{algorithm2e}

\makeatletter
\renewcommand{\@algocf@capt@plain}{above}
\makeatother

\usetikzlibrary{patterns}
\usetikzlibrary{decorations.pathreplacing}

\theoremstyle{definition}
\newtheorem{observation}[theorem]{Observation}

\newcommand{\floor}[1]{\left\lfloor #1 \right\rfloor}
\newcommand{\ceil}[1]{\left\lceil #1 \right\rceil}
\newcommand{\dotpr}[2]{\left\langle #1 , #2 \right\rangle}

\newcommand{\dist}[1]{\lVert #1 \rVert}
\newcommand{\prob}[1]{\mathbb{P}\left[ #1 \right]}

\newcommand{\CNN}[2]{\textsf{NN}_{\mathrm{wfn}}(#1, #2)}
\newcommand{\CNNE}{$\textsf{NN}_{\mathrm{wfn}}$}
\newcommand{\pre}[2]{preproc(#1, #2)}
\newcommand{\que}[2]{query(#1, #2)}
\newcommand{\Oc}{\mathcal{O}}
\newcommand{\Rspace}{\mathbb{R}}
\newcommand{\Rdspace}{\mathbb{R}^d}
\newcommand{\Rkspace}{\mathbb{R}^k}

\newcommand{\spheret}{\mathbb{S}^{(d-1)}}

\newcommand{\ml}{$max$-$l_2$}
\newcommand{\mlp}{$max$-$l_2\_NN$}

\newcommand{\pfp}{\mbox{$\textsf{p}_{\textsf{fp}}$}}
\newcommand{\tpfp}{\mbox{$\tilde{\textsf{p}}_{\textsf{fp}}$}}

\title{Approximate nearest neighbors search without false negatives for $l_2$ for $c>\sqrt{\log\log{n}}$}

\author{Piotr Sankowski}
\author{Piotr Wygocki}
\affil{University of Warsaw, Poland\\
  \texttt{[sank,wygos]@mimuw.edu.pl}}

\authorrunning{P. Sankowski and P. Wygocki} 

\Copyright{Piotr Sankowski and Piotr Wygocki}

\subjclass{F.2.2; G.3}
\keywords{locality sensitive hashing, approximate nearest neighbor search, high-dimensional, similarity search}%

\begin{document}

\maketitle

\begin{abstract}

In this paper, we report progress on answering the open problem presented by
Pagh~\cite{Pagh15}, who considered the nearest neighbor search without false
negatives for the Hamming distance. We show new data structures for solving the $c$-approximate nearest
neighbors problem without false negatives for Euclidean high dimensional space
$\mathcal{R}^d$. These data structures work for any $c = \omega(\sqrt{\log{\log{n}}})$,
where $n$ is the number of points in the input set, with poly-logarithmic query time and polynomial
pre-processing time. This improves over the known algorithms, which require $c$ to be
$\Omega(\sqrt{d})$. 

This improvement is obtained by applying a sequence of reductions, which are interesting on their own.
First, we reduce the problem to $d$ instances of
dimension logarithmic in $n$. Next, these instances are reduced to a number of
$c$-approximate nearest neighbor search without false negatives instances in
$\big(\Rspace^k\big)^L$ space equipped with metric $m(x,y) = \max_{1 \le i \le
L}(\dist{x_i - y_i}_2)$. 

\end{abstract}

\section{Introduction}

The nearest neighbor search has numerous applications in image processing, search engines, recommendation engines, predictions and machine learning.
We define the nearest neighbor problem as follows: for a given input set, a query point and a distance $R$,  return a point (optionally all points) from the input set, which is closer to the query point than  $R$ in the given metric (typically $l_p$ for $p\in[1,\infty]$), or report that such a point does not exist.
The input set and the distance $R$ are known in advance. Hence, the input set may be preprocessed what may result in reducing the query time.
The problem in which the distance $R$ is not known during the preprocessing and our task is to find the nearest neighbor can be efficiently reduced to the problem defined as above \cite{motwani}.\footnote{
Authors used to distinguish between these two problems. The problem in which the radius is known in pre-processing is sometimes called Point Location in Equal Balls (PLEB) \cite{motwani}.}
Unfortunately, the nearest neighbors search, defined as above, appears to be intractable for high dimensional spaces such as $l_p^d$ for large~$d$.
The existence of an algorithm with a sub-linear in the data size and not exponential in $d$ query time and with not exponential in $d$ pre-processing,
would contradict the strong exponential time hypothesis~\cite{Williams2005}.
In order to overcome this obstacle, the $c$-approximate nearest neighbors problem with $c>1$, was introduced.
In this problem, the query result is allowed to contain points which are within the distance $cR$ from the query point.
In other words, the points within the distance $R$ from the query point are classified as neighbors, the points farther than $cR$ are classified as non-neighbors,
while the rest may be classified into any of these two categories.
This assumption makes the problem easier, for many  metric spaces such as $l_p$ when $p \in [1,2]$ or the Hamming space \cite{motwani}.
On one hand, sub linear in the input size queries are possible.  On the other hand, the queries and pre-processing times are 
polynomial in the dimension of the space. 

Locality sensitive hashing (LSH) is one of the major techniques for solving the \mbox{$c$-approximate} nearest neighbor search.
Many LSH functions are mappings which roughly preserve distances.
A random LSH function maps two 'close' points to two 'close' hashes with 'large' probability.
Analogously, two 'distant' points  are mapped to two 'distant' hashes with 'large' probability.
Roughly speaking, the LSH is used to reduce the dimension of the input space, which allows to solve the problem in the lower dimensional space. Thus, the efficiency of the algorithm
strongly depends on the quality of LSH functions used. The crucial properties of the LSH functions are the probability
of false positives and the probability of false negatives.
A false negative is a point which is 'close' to the query point, but its hash is 'far away' from the hash of the query point.
Analogously, the false positive is a point whose distance to the query point is 'large', but it is mapped to a 'close' hash.

The previously known algorithms for the $c$-approximate nearest neighbors (see e.g., \cite{DBLP:journals/cacm/AndoniI08,CHAZELLE200824}) give Monte Carlo guaranties
for returned points, i.e., an input point close to the query point is returned with some probability.
In other words, there might be false negatives.
For example, a common choice of the hash functions is   $f(x) = \dotpr{x}{v}$ or $f(x) = \floor{\dotpr{x}{v}}$,
where $v$ is a vector of numbers drawn independently from some probability distribution \cite{DBLP:journals/cacm/AndoniI08,motwani,pacukatall}.
For a Gaussian distribution, $\dotpr{x}{v}$ is also Gaussian with zero mean  and standard deviation equal to $\dist{x}_2$.
It is easy to see that these are LSH functions for $l_2$, but as mentioned above, they only give probabilistic guaranties.
In this paper, we aim to enhance this by focusing on the $c$-approximate nearest neighbor search without false negatives for $l_2$.
In other words, we consider algorithms, where a point 'close' to the query point is guaranteed to be returned.


Throughout this paper, we assume that $n \gg d$ and $\exp(d) \gg n$. This represents a situation where the exhaustive scan through all the input points, as well as the
usage of data structures exponentially dependant on $d$, become intractable. The typical values to consider could be $n = 10^9$ and $d=100$. If not explicitly specified, all statements assume the usage of the $l_2$ norm.

\section{Related Work}

There exists an efficient, Monte Carlo $c$-nearest neighbor algorithm for $l_1$ \cite{motwani} with the query 
and the pre-processing complexity equal to $\Oc(n^{1/c})$ and $\Oc(n^{1+1/c})$, respectively.
For $l_2$ in turn, there exists a near to optimal \cite{O'Donnell:2014:OLB:2600088.2578221} algorithm \cite{DBLP:journals/cacm/AndoniI08}
with the query and the pre-processing complexity equal to $n^{1/c^2+ o(1)}$ and $n^{1+1/c^2 + o(1)}$, respectively.
Moreover, the algorithms presented in \cite{motwani} work for $l_p$ for any $p\in [1,2]$.
There are also data dependent algorithms, which take into account the actual distribution of the input set \cite{data-depended-hashing},
which achieve query time  $dn^{\rho+o(1)}$ and space $\Oc(n^{1+\rho+o(1)}+dn)$, where $\rho=1/(2c^2-1)$.


Pagh~\cite{Pagh15} considered the $c$-approximate nearest neighbor search without false negatives~(\CNNE{}) for the Hamming space, obtaining the results close to those of \cite{motwani}.
He showed that the bounds of his algorithm for $cR = \log( n/k)$ differ by at most a factor of $\ln 4$ in
the exponent in comparison to the bounds of \cite{motwani}. 
Recently, Ahle~\cite{DBLP:journals/corr/Ahle17} showed an optimal \cite{O'Donnell:2014:OLB:2600088.2578221} algorithm for the nearest neighbor without false negatives for Hamming space
and  Braun-Blanquet metric.
Indyk~\cite{DBLP:conf/focs/Indyk98} provided a deterministic algorithm
for $l_{\infty}$ for $c = \Theta(\log_{1+\rho}\log{d})$ with the storage $\Oc(n^{1+\rho}\log^{O(1)}n)$ and the query time  $\log^{O(1)}n$ for some tunable parameter $\rho$. He proved
that the nearest neighbor without false negatives for  $l_{\infty}$ for $c<3$ is as hard as the subset query problem, a long-standing combinatorial problem. 
This indicates that the nearest neighbor without false negatives for  $l_{\infty}$ might be hard to solve
for any $c>1$.
Also, Indyk~\cite{Indyk:2007:UPE:1250790.1250881} considered deterministic mappings $l_2^n \rightarrow l_1^m$, for $m=n^{1+o(1)}$, which might be useful for constructing
efficient algorithms for the  nearest neighbor without false negatives \cite{Pagh15}.

Pacuk~et~al.~\cite{pacukatall} presented a schema for solving the  nearest neighbor without false negatives for any $p\in [1,\infty]$ for $c = \Omega(d^{\max\{1/2, 1-1/p\}})$.
Using the enhanced hash functions, Wygocki~\cite{wygos} presented algorithms with improved complexities.
He considered two hashing families, giving different trade-offs between the execution times and the conditions on $c$.
In particular, (Theorem 3, case  2, for $p=2$ in~\cite{wygos}):

\begin{theorem}[\cite{wygos}]\label{main}
For any $c >  \tilde \tau = 2\sqrt{d}$, there are data structures for the  nearest neighbor without false negatives  with
\begin{itemize}
    \item $\Oc(n^{1+\frac{\ln{3}}{\ln(c/\tilde  \tau)}})$ pre-processing time  and $\Oc(d|P| + d\log{n} + d^2)$ query time for the 'fast query' algorithm,
    \item $\Oc(n d \log{n})$ pre-processing time  and $\Oc(d(|P| +n^{\frac{\ln{3}}{\ln(3c/ \tilde  \tau)}}))$ query time for the 'fast pre-processing' algorithm,
\end{itemize}
where $|P|$ is the size of the result.
\end{theorem}

The dimension reduction with means of random linear mappings was considered previously in a more general context. 
In particular, Johnson-Lindenstrauss Lemma \cite{johnson84extensionslipschitz} is the most well known 
reference here. The concentration bounds used to prove
this classic result will be very useful in our reductions:

\begin{lemma}[Johnson-Lindenstrauss]\label{jl1}
    Let $Y\in \Rspace^d$ be  chosen  uniformly  from  the  surface  of  the $d$-dimensional  sphere.   Let
    $Z=(Y_1,Y_2,\dots,Y_k)$ be  the  projection  onto  the  first $k$ coordinates,  where $k < d$.   Then  for  any
    $\alpha < 1$
    :
    \begin{equation}
    \prob{\frac{d}{k}\dist{Z}_2^2 \le \alpha} \le \exp(\frac{k}{2}(1-\alpha+\log{\alpha})),
    \end{equation}

\end{lemma}

\section{Our contribution}
Recently,  efficient algorithms were proposed for solving the approximate nearest neighbor search without false negatives
for $c = \Omega( \max\{\sqrt{d}, d^{1-1/p}\})$ in $l_p$ for any $p \in [1,\infty]$ \cite{pacukatall,wygos}.
The main problem with these algorithms is the constraint on $c$. For $l_2$, the previous result require $c$ to be of order of $\Omega(\sqrt{d})$,
thus the nearest neighbor algorithm were allowed to return points within  $\Oc(\sqrt{d}R)$ radius from query point.
We relax this to any $c$, which makes the presented algorithms usable in practical cases.
The contribution of this paper is relaxing this condition and improving the complexity of the algorithms for $l_2$:
\begin{itemize}

    \item We show that the \CNNE{} can be reduced to $d$
        instances of \CNNE{} in $\mathbb{R}^{\log{n}}$. For our typical settings of
        parameters, the factor of $d$ is negligible. As a result, reducing
        the dimension leads to reducing the complexity of the
        problem. Moreover, it leads to relaxing the conditions on $c$ to $c=\Omega(\sqrt{\log n})$.

    \item Further reductions lead to algorithms for any $c = \omega(\sqrt{\log{\log{n}}})$. We introduce an
        algorithm with the $n^{o(1)}$ query time and polynomial pre-processing time, which for large $c$ tents to $n^{1+o(1)}$.

\end{itemize}

The first reduction is interesting on its own since further work on the problem can be done under the assumption that the dimension of the problem is logarithmic in $n$.
This simplifies the problem at a cost of multiplying the complexities by a factor of $d$.
Also, the authors of \cite{pacukatall} proved that their construction is tight for $d = \omega(\log{n})$, living the case of  $d = \Theta(\log{n})$ inconclusive.

\subsection{Used Methods}

In order to relax the conditions on $c$, we apply a sequence of dimension reductions.
In Section~\ref{lsh_sec}, we show how to reduce the $c$-approximate nearest neighbors in $l_2^d$ ($\CNN{c}{d}$) to $d/\log(n)$
instances of $\CNN{\Oc(c)}{\Oc(\log{n})}$.  Applying the algorithm of \cite{wygos} as a
black box gives the first improvement over \cite{wygos}: an efficient algorithm
for $c = \Omega(\sqrt{\log{n}})$.  The reduction is based on the well-known
Johnson-Lindenstrauss Lemma~\cite{johnson84extensionslipschitz}. We introduce $d/\log(n)$ linear mappings,
each reduces the dimension of the original problem.  Each mapping
roughly preserves the length of the vector and additionally at least one of them does not increase the length of the input vector. The property of not
increasing the length of the vector is crucial. For two 'close' vectors $x,y \in
\Rspace^d$: $\dist{x-y}_2 < 1$ and a linear mapping $A$, $Ax$ and $Ay$ are
'close' if and only if $\dist{Ax - Ay}_2 = \dist{A(x-y)}_2 < 1$ , so $A$ maps a
'small' vector $x-y$, to a 'small' vector $A(x-y)$.

In Section \ref{algo2}, we show further reductions, which enable us to relax
the constraint to $c = \omega(\sqrt{\log{\log{n}}})$. We extend the reduction from
Section \ref{lsh_sec} by using a number of mapping families. This leads to an
interesting sub-problem of solving the approximate nearest neighbors in
$(\Rspace^k)^L$, for norm \ml{}$(x) \coloneqq \max_{1 \le i \le L}\dist{x_i}_2$ and the
induced metric. This norm is present in literature and was denoted as max--product  or $l_{\infty}$--product.
Apparently, the $c$-approximate nearest neighbor search in \ml{} might be solved using the LSH functions family introduced in \cite{wygos}.

This series of reductions leads to our final results. First we reduce the problem 
to a number of $\CNN{\Oc(c)}{\Oc(\log{n})}$ instances, each of which is further
reduced to a number of problems in \ml{}, which in turn are solved using the
LSH functions presented in \cite{wygos}.

\section{Notations}

The $c$-approximate nearest neighbors search without false negatives with parameter $c
> 1$ and the dimension of the space equal to $d$, will be denoted as
$\CNN{c}{d}$.  The expected query and pre-processing time complexities of
$\CNN{c}{d}$ will be denoted as $\que{c}{d}$ and $\pre{c}{d}$ respectively.
The input set will be denoted as $X$ and it will always be assumed to contain $n$
points. W.l.o.g, throughout this work we will assume, that $R$ -- a given radius
equals 1 (otherwise, all vectors' lengths might be rescaled by $1/R$). The $\tilde \Oc()$ denotes the complexity up
to the poly logarithmic factors i.e., $\tilde \Oc(f(n)) = \Oc(f(n)poly(log(n)))$.
$\dist{\cdot}_2$ denotes the standard norm in $l_2$, i.e., $\dist{x}_2 = (\sum_i{|x_i|^2})^{1/2}$.
The $f(n) = \omega(g(n))$ means that  $f$  dominates $g$  asymptotically, i.e., $g(n) = o(f(n))$.

\section{Algorithm for \texorpdfstring{$\boldsymbol{c=\Omega(\sqrt{\log n})}$}{Lg}} \label{lsh_sec}

The basic idea is the following: we will introduce a number of linear mappings
to transform the $d$-dimensional problem to a number of problems with
dimension reduced to $\Oc(\log{n})$. Then we use
the algorithm introduced in \cite{wygos}, to solve these problems in the space with the reduced
dimension.

We will introduce $d/k$\footnote{For simplicity, let us assume that $k$ divides $d$, 
this can be achieved by padding extra dimensions with $0$'s.}
linear mappings $A^{(1)}, A^{(2)}, \dots, A^{(d/k)}:R^d \rightarrow R^k$, where $k<d$ and
show the following properties:
\begin{enumerate}
    \item for each point $x \in \Rdspace$, such that $\dist{x}_2 \le 1$, there exists $1 \le i \le d/k$, such that $\dist{A^{(i)} x}_2 \le 1$,
    \item for each point $x \in \Rdspace$, such that $\dist{x}_2 \ge c$, where $c>1$, the probability that there exists $1 \le i \le d/k$, such that $\dist{A^{(i)}x}_2 \le 1$ is bounded.
\end{enumerate}

The property $1.$ states, that for a given 'short' vector (with a length
smaller than $1$), there is always at least one mapping, which transforms this
vector to a vector of length smaller than $1$.  Moreover, we will show, that
there exists at least one mapping $A^{(i)}$, which does not increase the length
of the vector, i.e., such that $\dist{A^{(i)} x}_2 \le \dist{x}_2$. The
property $2.$ states, that we can bound the probability of a 'long' vector
($\dist{x}_2 > c$), being mapped to a 'short' one ($\dist{A^{(i)} x}_2 \le
1$). Using the standard concentration measure arguments, we will prove that
this probability decays exponentially in $k$.

\subsection{Linear mappings}

In this section, we will introduce linear mappings satisfying properties 1. and 2.
Our technique will depend on the concentration bound used to
prove the classic Johnson-Lindenstrauss Lemma.
In Lemma \ref{jl1},  we take a random vector and project it to
the first $k$ vectors of the standard basis of $\Rspace^d$.  In our settings,
we will project the given vector to a random orthonormal basis which gives the
same guaranties.
The mapping  $A^{(i)}$ consists
of $k$ consecutive vectors from the random basis of the
$\Rspace^d$ space scaled by $\sqrt{\frac{d}{k}}$.
The following reduction describes the basic properties of
our construction:

\begin{lemma}[Reduction Lemma]\label{rjl}
For any parameter $\alpha \ge 1$ and $k < d$, there exist $d/k$ 
linear mappings $A^{(1)}, A^{(2)}, \dots, A^{(d/k)}$, from $\Rdspace$ to $\Rkspace$, such that:
\begin{enumerate}
    \item for each point $x\in \Rdspace$ such that $\dist{x}_2 \le 1$, there exists $1 \le i \le d/k$, such that $\dist{A^{(i)} x}_2 \le 1$,

    \item for each point $x \in \Rdspace$ such that $\dist{x}_2 \ge c$, where $c>1$,
        for each $i$: $1 \le i \le d/k$, we have
        \[\prob{\dist{A^{(i)}x}_2 \le \alpha} < e^{-k(\frac{c-\alpha}{2c})^2}.\]

\end{enumerate}
\end{lemma}
\begin{proof}

Let $a_1, a_2, \dots, a_d$ be a random basis of $R^d$. Each of the
$A^{(i)}$ mappings is represented by a $k \times d$ dimensional matrix.  We will use $A^{(i)}$ for denoting both the mapping and the corresponding matrix.
The $j$th row of the matrix $A^{(i)}$ equals $A^{(i)}_j = \sqrt{\frac{d}{k}}
a_{(i-1)k+j}$. In other words, the rows of $A^{(i)}$ consist
of $k$ consecutive vectors from the random basis of the
$\Rspace^d$ space scaled by $\sqrt{\frac{d}{k}}$.

To prove the first property, observe that $A = \sum_{i=1}^d \dotpr{a_i}{x}^2 \le
1$, since the distance is independent of the basis.  Assume on the contrary, that for each $i$,
$\dist{A^{(i)}_2 x} > 1$. It follows that $d \ge d A = k \sum_{i=1}^{d}\dist{A^{(i)} x}_2^2 > d$.
This contradiction ends the proof of the first property.

For any $x \in \Rspace^d$, such that $\dist{x}_2 > c$, the probability:

\[\prob{\dist{A^{(i)}x}_2 \le \alpha} =
  \prob{\frac{\dist{A^{(i)}x}_2^2}{c^2} \le (\frac{\alpha}{c})^2} \le
  \prob{\frac{\dist{A^{(i)}x}_2^2}{\dist{x}_2^2} \le (\frac{\alpha}{c})^2}. \]
Using the fact that $\log x < x-1  - (x-1)^2/2$ for $x < 1$ and Lemma \ref{jl1}, the above is bounded as follows:

\[\prob{\frac{\dist{A^{(i)}x}_2^2}{\dist{x}_2^2} \le (\frac{\alpha}{c})^2} \le
 \exp\big(-\frac{k}{4}(1-(\frac{\alpha}{c})^2)^2\big) \le e^{-k(\frac{c-\alpha}{2c})^2},\]
which completes the proof.
\end{proof}

\subsection{Algorithm}\label{algo}

The algorithm works as follows: for each $i$, we project $\Rspace^d$ to
$\Rspace^k$ using $A_i$ and solve the corresponding problem in the smaller
space. For each query point, we need to merge the solutions obtained for each
subproblem.  This results in reducing the $\CNN{c}{d}$ to $d/k$ instances of
$\CNN{\alpha}{k}$.

\begin{lemma}\label{lemalgo}
For $1 < \alpha < c$ and $k < d$, the $\CNN{c}{d}$ can be reduced to $d/k$ instances of the $\CNN{\alpha}{k}$.
The expected pre-processing time equals $\Oc(d^2 n + d/k\ \pre{\alpha}{k})$ and
the expected query time equals $\Oc(d^2 + d/k\ e^{-k(\frac{c-\alpha}{2c})^2} n + d/k\ \que{k}{\alpha})$.
\end{lemma}
\begin{proof}
We use the assumption that $k < d < n$ to simplify the complexities.
The pre-processing time consists of:
\begin{itemize}
    \item $d^3$: the time of computing a random orthonormal basis of $\Rspace^d$.
    \item $d^2 n$: the time of changing the basis to $a_1, a_2, \dots, a_d$.
    \item $d n k$: the time of computing $A^{(i)} x$ for all $1 \le i \le d$ and for all $n$ points.
    \item $d/k\ preproc(\alpha, k)$: the expected pre-processing time of all subproblems.

\end{itemize}
The query time consists of:
\begin{itemize}
    \item $d^2$: the time of changing the basis to $a_1, a_2, \dots, a_d$.
    \item $d/k\ e^{-k(\frac{c-\alpha}{2c})^2} n$: the expected number of false positives (by Lemma \ref{rjl}).
    \item $d/k\ query(k, \alpha)$: the expected query time of all subproblems.
\end{itemize}
\end{proof}

The following corollary simplifies the formulas used in Lemma \ref{lemalgo} and
shows that if $\frac{c}{c-\alpha}$ is bounded, the $\CNN{c}{d}$ can be reduced
to a number of problems of dimension $\log{n}$ in an efficient way.  Namely, setting
$k=\big(\frac{2c}{c-\alpha}\big)^2\log{n}$ we get:

\begin{corollary}\label{colr}

    For any $1 \le \alpha < c$ and $\gamma\log{n} < d$, the $\CNN{c}{d}$ can be
    reduced to $d/\log{n}$ instances of the $\CNN{\alpha}{\gamma\log{n}}$, where $\gamma =
    \big(\frac{2c}{c-\alpha}\big)^2$ and:

    \[\que{c}{d} = \Oc(d^2 +  d/\log(n)\ \que{\alpha}{\gamma\log{n}}),\]
   \[\pre{c}{d} =\Oc(d^2 n  + d/\log(n)\ \pre{\alpha}{\gamma\log{n}}).\]
\end{corollary}

Combining the above corollary with the results introduced in \cite{wygos}, we can achieve
the algorithm with the polynomial pre-processing time and the sub-linear query time.
Theorem \ref{main} states, that for any $c > 2\sqrt{d} $, the
$\CNN{c}{d}$ can be solved in
the $\Oc(n^{1+\frac{\log{3}}{\log(c/\tilde\tau)}})$ pre-processing time  and the query time
equal to $\Oc(d|P| + d \log{n} + d^2 )$, where $P$ is the size of the result set and $\tilde\tau = 2\sqrt{d}$.
Altogether, setting $\alpha = c/2$ in Corollary \ref{colr}, we get\footnote{
The author of \cite{wygos} presented multiple algorithms giving different trade-offs between the pre-processing time and the query time.
Particularly, the algorithm with the $\Oc(n\log{n})$ processing time and the sub-linear query time was presented.
The same can be done for Theorem \ref{th1}. We omit this to avoid the unnecessary complexity.}:
\begin{theorem}\label{th1}
    The $\CNN{c}{d}$ can be solved for any $c > \tilde \kappa = 16\sqrt{\log{n}}$ with:
    \[\que{c}{d} = \Oc(d |P| + d \log{n} + d^2)\]
\[\pre{c}{d} =\Oc(d n^{1+\frac{\ln{3}}{\log(c/\tilde \kappa)}}/\log(n) ).\]
\end{theorem}
The time complexity of the algorithm is the same as for $c=\Omega(\sqrt{d})$, the pre-processing time is larger by a factor of $d/\log(n)$.

\section{The algorithm for \texorpdfstring{$\boldsymbol{c=\omega(\sqrt{\log(\log(n))})}$}{Lg}}\label{algo2}

In this section we give another algorithm which works for $c=\omega(\sqrt{\log(\log(n))})$.
Lemma~\ref{jl1} implies that the $\CNN{c}{d}$ problem can be reduced to $d/\log(n)$
problems of dimension logarithmic in $n$. In order to reduce the dimension even
more, we will employ $L$ independent families of linear mappings introduced in
Section \ref{lsh_sec}.  In each of the families, there is at
least one mapping, which does not increase the length of the input vector.
As a result, there exists a combination of $L$ mappings (each mapping taken from
a distinct family) which do not increase the input vector length. Also, for any
combination of $L$ mappings, the probability that all the mappings
transform a 'long' vectors to a 'short; one can be bounded. The structure of the mappings
is presented in Figure \ref{fig:mappings}.

 \begin{figure*}[ht!]
\centering

\begin{tikzpicture}[scale = 0.5]

    \node (a11) at (0,-1) {$A^{(1,1)}$};
    \node (ai11)  at (0,-3) {$A^{(i_1,1)}$};
    \node (ad1)  at (0,-6) {$A^{(d/k,1)}$};
    \path (a11) -- (ai11) node [red, midway, sloped] {$\dots$};
    \path (ai11) -- (ad1) node [red, midway, sloped] {$\dots$};

    \node (a12) at (3,-1) {$A^{(1,2)}$};
    \node (ai22)  at (3,-4) {$A^{(i_2,2)}$};
    \node (ad2)  at (3,-6) {$A^{(d/k,2)}$};
    \path (a12) -- (ai22) node [red, midway, sloped] {$\dots$};
    \path (ai22) -- (ad2) node [red, midway, sloped] {$\dots$};

    \node (a1d) at (7,-1) {$A^{(1,L)}$};
    \node (aidd)  at (7,-4) {$A^{(i_d,L)}$};
    \node (add)  at (7,-6) {$A^{(d/k,L)}$};
    \path (a1d) -- (aidd) node [red, midway, sloped] {$\dots$};
    \path (aidd) -- (add) node [red, midway, sloped] {$\dots$};

    \node[red] (dot) at (5,-3) {$\dots$};
    \draw[blue,thick] (ai11) -- (ai22);
    \draw[blue,thick] (ai22) -- (dot);
    \draw[blue,thick] (dot) -- (aidd);

\end{tikzpicture}
\caption{Each column describes one family of linear mappings, constructed based on one random, orthonormal basis. The blue path describes one combination of mappings.
    }
\label{fig:mappings}
\end{figure*}
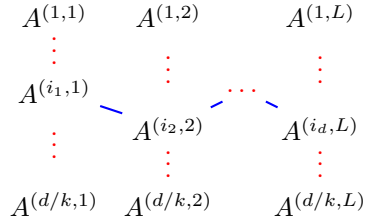

To formalize the above  line of thinking, we introduce the following lemma:
\begin{lemma}\label{lemll}
For any natural number $L>0$, there exist $d/k\ L$ linear mappings $A^{(i,j)}:R^d\rightarrow R^k$, where $k<d$, $1 \le i \le d/k$ and $1 \le j \le L$,  such that
\begin{enumerate}
    \item for each point $x\in R^d$ which satisfies $\dist{x}_2 \le 1$, there exist 
      $1 \le i_1, i_2, \dots, i_L \le d$ such that $\dist{A^{(i_j,j)} x}_2 \le 1$, for each $1 \le j\le L$.

    \item for each point $x \in R^d$ which satisfies $\dist{x}_2 \ge c$, where $c>1$,
        for each $i_1, i_2, \dots, i_L$: $1 \le i_1, i_2, \dots, i_L \le d/k$, we have
        \[\prob{\forall_j:\dist{A^{(i_j,j)}x}_2 \le \alpha} < \exp\big(-\frac{kL}{4}(\frac{c-\alpha}{c})^2\big).\]
\end{enumerate}
\end{lemma}

\begin{proof}
For each $j$: $1 \le j \le L$ we independently sample the orthonormal basis of $\Rspace^d$:
$a_1,a_2, \dots, a_d$.  The $A^{(i,j)}$ will be created in the same way as in
Lemma \ref{rjl}, namely, the $t$--th row of $A^{(i,j)}$ equals $A^{(i,j)}_t =
\sqrt{\frac{d}{k}}a_{(i-1)k+t}$ . The properties (1) and (2) follow directly from Lemma \ref{rjl}.

\end{proof}

In order to employ Lemma \ref{lemll} for a given query $q$,
we need to be able to find all points in $X$ such that
a given combination of mappings transforms these points and the query point to
'close' vectors.
In other words, we need to find all $c$-approximate nearest neighbors for the transformed input set
 $\tilde X \subset (\Rspace^k)^L$ in the space equipped with metric: \ml{}$(x,y) = \max_{1 \le i \le
L}(\dist{x_i - y_i})$, which is formally defined as follows:

\begin{definition}[the $c$-approximate nearest neighbor search in \ml{}]

    The \mlp$(c, L, k)$ is defined as follows: given  a query point $q \in
    (\Rspace^k)^L$ and a set $\tilde X \subset (\Rspace^k)^L$ of $n$ input
    points, find all input points, such that for each $1 \le i \le L$:
    $\dist{q_i-\tilde x_{i}}_2 \le 1$. Moreover, each $\tilde x$ satisfying $\forall_i
    \dist{q_i-\tilde x_{i}}_2 \le c$, might be returned as well.  Finally, each
    $x$ such  that $\exists_i \dist{q_i-\tilde x_{i}} > c$, must not be returned.

\end{definition}

Using the construction from Lemma \ref{lemll}, the $\CNN{c}{d}$ problem can be
reduced to $d^L$ instances of the \mlp{}$(\alpha, L, k)$.  Each of the
instances is represented by indices: $\{i_1, i_2, \dots, i_L\} $ and the
corresponding mappings $A^{(i_j,j)}$ for $1 \le j \le L$.
Each input point $\tilde x \in \tilde X$ comes from the point $x \in X$ by
applying the mappings:  $\tilde x = (A^{(i_1,1)}x,\dots,  A^{(i_L, L)}x )$.  Similarly,
the query point $q$, in \ml{}, is created from the query point $q$ in
$\CNN{c}{d}$, as $(A^{(i_1,1)}q ,\dots,  A^{(i_L, L)}q )$.

\begin{lemma}\label{lemall}
The  $\CNN{c}{d}$ can be reduced to $(d/k)^L$ instances of \mlp{}$(\alpha, L,k)$.
The expected pre-processing time equals:
\[\Oc(L d^2 n + (d/k)^L\ preproc_{\text{\ml{}}}(\alpha, L, k))\]
and the expected query time equals:
\[\Oc(L d^2 + (d/k)^L e^{-kL(\frac{c-\alpha}{2c})^2} n + (d/k)^L\ query_{\text{\ml{}}}(\alpha, L, k)).\]
\end{lemma}
The proof of the Lemma is analogical to the proof of Lemma \ref{lemalgo}. The
following corollary presents the simplified version of Lemma \ref{lemall}.
Setting $k=\lceil L^{-1}\big(\frac{2c}{c-\alpha}\big)^2\log{n} \rceil$ we get:

\begin{corollary}\label{colr2}
    For any $1 \le \alpha < c$, the $\CNN{c}{d}$ can be reduced to $d^L$ instances of \mlp{}$(\alpha, L, L^{-1}\gamma\log{n})$, where $\gamma = \big(\frac{2c}{c-\alpha}\big)^2$ and:
    \[\que{c}{d} = \Oc(L d^2  + (d/\log(n))^L\ query_{\text{\ml{}}}(\alpha, L, \lceil L^{-1}\gamma\log{n} \rceil)),\]
    \[\pre{c}{d} =\Oc( L d^2 n  + (d/\log(n))^L\ preproc_{\text{\ml{}}}(\alpha, L, \lceil L^{-1}\gamma\log{n}\rceil)).\]
\end{corollary}

The \mlp{}$(\alpha, L, k)$ can be trivially solved by dealing with each of the $L$-dimensional $\CNN{\alpha}{k}$ problems separately. Unfortunately, this 
gives unacceptable complexities. In order to improve complexity of algorithm for the $\CNN{c}{d}$ problem, we need
to be able to solve the \mlp{} more efficiently.

\subsection{Solving the \texorpdfstring{$\boldsymbol{c}$}{Lg}-approximate nearest neighbors in
\texorpdfstring{$\boldsymbol{max}$}{Lg}-\texorpdfstring{$\boldsymbol{l_2}$}{Lg}}

 In order to solve this problem, we use the standard LSH technique
    based on the hash functions $\tilde h$  introduced in \cite{wygos}  defined as follows:
\begin{displaymath}
    \tilde h(x) = \floor{\dotpr{w}{x}}, \text{ where }w \text{ is a random vector from the unit sphere }\spheret.
\end{displaymath}

    We consider two hashes to be 'close' if $|\tilde h(x) - \tilde h(x')| \le 1$.
    Based on $\tilde h$, we introduce a new hash function $g$. Each of
    the input points is hashed by $g$ and the reference to this point is kept in a
    single hash map. For a given query point, we examine all input points
    which are hashed to the same value as the query point.

    Namely, each $\tilde x \in (\Rspace^k)^L$ will be hashed by
    $g(\tilde x) \coloneqq (g_1(\tilde x_1), \dots, g_L(\tilde x_L))$, where
    $g_i(x)\coloneqq(\tilde h_1(x), \tilde h_2(x), \ldots, \tilde h_w(x))$ is a hash function
    defined as a concatenation of $w$ random LSH functions $\tilde h$.
    The function $g$ can be also seen as a concatenation of $w L$ random hash functions~$\tilde h$.
    If two points are 'close' in the considered \ml{} metric, then $g$ transforms these points
    to hashes $p^{(1)}, p^{(2)} \in \mathbb{Z}^{w L}$,  such that
    $|p^{(1)}_i - p^{(2)}_i| \le 1$ for all $i \in {w L}$. The pre-processing algorithm is summarized in the following pseudocode:

\vspace{0.5cm}
\begin{algorithm}[H]
 \KwData{$X \subset (\Rspace^k)^L$ - the set of $n$ input points}
 \KwResult{$H:\mathbb{Z}^{wL} \rightarrow 2^X$ - the hash map storing for each hash $\alpha \in \mathbb{Z}^{wL}$ the subset of input points with hashes close to $\alpha$ }
 $H=\emptyset$\;

 \For{$x \in X$}{
  $\alpha = g(x)$\;
  \For{$\alpha'$ such that $\dist{\alpha -\alpha'}_{\infty}\le 1$}{
   $H(\alpha').push(x)$\;
   }
 }
 \caption{The pre-processing algorithm}
\end{algorithm}
\vspace{0.5cm}

\noindent The query algorithm consists of examining the bucket for $g(guery\_point)$:

\vspace{0.5cm}
    \begin{algorithm}[H]
 \KwData{$q \in (\Rspace^k)^L$ - the query point}
 \KwResult{$P \subset X$ - the set of neighbors of $q$}
 $P = \emptyset$\;

 \For{$x \in H(g(q))$}{
   \If{$x$  is a neighbor of $q$}{
     $P.push(x)$\;
   }
 }
 \caption{The query algorithm}
\end{algorithm}
\vspace{0.5cm}

The following theorem describes the above algorithm:

\begin{theorem}\label{ml}
    For $L = o(\log{n})$ and $c > 2\sqrt{k}$, the \mlp{}(c,L,k) can be solved in the $\Oc(k L|P| + k \ln{n} + Lk^2)$ query time and
    $\Oc(n^{1 + \frac{\ln{3}}{\ln(c/\tilde \kappa)}})$ pre-processing time complexity for $\tilde \kappa = 2\sqrt{k}$.\footnote{Theorem \ref{ml} might be generalized, to any $p\in[1,\infty]$.
The generalization is done by applying hash functions suited for $l_p$. Such hash functions where introduced in \cite{wygos}.}
\end{theorem}
\begin{proof}

Let us start with the key properties of the LSH family.

\begin{observation}['Close' points have 'close' hashes for $\tilde h$ (Observation 5 in \cite{wygos})] \label{small}
For $x,y \in \Rdspace$, if $\dist{x-y} < 1$ then $\forall_{\tilde h} |\tilde h(x) - \tilde h(y)| \le 1$.
\end{observation}

\begin{lemma}[The probability of false positives for $\tilde h$ (Lemma 2 in \cite{wygos})]\label{big_ball}

For $x,y \in \Rdspace$ and $c > \tilde \tau = 2d^{1/2}$ such that $\dist{x-y} > c$, it holds:

\begin{displaymath}
    \tpfp = \prob{|\tilde h(x) - \tilde h(y)| \le 1} < \tilde \tau / c.
\end{displaymath}
\end{lemma}

    Since we consider two hashes to be 'close', when they differ at most by one (see
     Observation~\ref{small}), for each hash
     $\alpha \in \mathbb{Z}^{wL}$ we need to store the reference to every point, that satisfies
     $\dist{\alpha - g(x)}_{\infty} \le 1$. Thus, the hash map size is $\Oc(n3^{wL})$.
     Computing a single $\tilde h$ function in $\Rspace^k$ takes $\Oc(k)$, so
     evaluating the $g(x)$ for $x \in (\Rspace^k)^L$ takes $O(w k L)$.
     The pre-processing consists of computing the $3^{wL}$ hashes for each point in the input set.
     The query consists of computing the hash of the query point, looking up all the points with colliding hashes,
     filtering out the false positives and returning the neighbors.

     It is easy to derive the following complexities:

For any $c>2k^{1/2}$ and the number of iterations $w \ge 1$, there exists a
 \mlp(c,L,k) algorithm with the following properties:
 \begin{itemize}
     \item the pre-processing time: $\Oc(n (w k L + 3^{w L}))$, where $w k L$ is the time needed to compute the $g(input\_point)$ and the $\Oc(3^{wL})$ is the number of the updated hashes for one input point,
     \item the expected query time: $\Oc(k L(|P| + w + n \tpfp^{w L}))$,
     where $w k L$ is the time needed to compute the $g(query\_point)$, $n \tpfp^{w L}$ is the number of false positives which need to be ignored, $|P|$ denotes the size of the result set.
     For each of the candidates,  we need to perform a check of complexity $\Oc(k L)$ to classify the point as a true positive or a false positive.

 \end{itemize}
 Above $\tpfp = \tilde \tau / c$ (see Lemma~\ref{big_ball}).

The number of iterations $w$ can be chosen arbitrarily, so we will choose the
optimal value. Denote $a = -\ln{\tpfp}$ and $b=\ln{3}$, then set $w$ to be:

\begin{displaymath}
    w = \ceil{ \frac{\ln{\frac{n a}{k} }}{a}L^{-1}  }
    .
\end{displaymath}

Let us assume that $n$ is large enough so that $w \ge 1$. Then, using the fact that
$x^{1/x}$ is bounded for $x > 0$ we have:

\begin{displaymath}
    3^{wL} \le 3 \cdot (3^{\ln{\frac{n a}{k}}})^{1/a} = 3 \cdot \big(\frac{n a}{k}\big)^{b/a} = 3 \cdot \big(\frac{n}{k}\big)^{b/a},
\end{displaymath}

\begin{displaymath}
    n \tpfp^{wL} = n e^{-a wL} \le n e^{-a \frac{\ln(\frac{n a}{k})}{a}}= \frac{k}{a}
    .
\end{displaymath}

Hence, for constant $c$ the expected query time is $\Oc(k L |P| + k \ln{n} + L k^2)$.
Subsequently, the pre-processing time is: $\Oc(n 3^{wL}) = \Oc(n^{1+b/a})$.
Substituting $a$, $b$ and $\pfp$ values gives the needed complexity guaranties.
\end{proof}

\subsection{Putting it All Together}

In order to achieve an efficient algorithm for  $c = \omega(\log(\log(n)))$, we will make a series of
reductions. First, using Corollary \ref{colr}, we reduce our problem to a number
of $\CNN{\Oc(c)}{\Oc(\log{n})}$ problems. Next, these problems are reduced to a number of \mlp{} problems
with dimension $k$ of $O(\log{\log{n}})$.  In the end, we use Theorem \ref{ml} to
solve the \mlp{}.

\begin{theorem}\label{th_log_log}
    The $\CNN{c}{d}$ can be solved with:
    \begin{itemize}
        \item the pre-processing time
    $\tilde{\Oc}(d^2 n  + d n^{1 + \frac{\ln{3}}{\ln(c/\mu)}+ 1/f(n)})$,
\item the query time
$ \tilde{\Oc}(d^2 + d n^{1/f(n)} |P|)$,
\end{itemize}
    for any $c > \mu =  D\sqrt{f(n)\log{\log{n}}}$,

    where $f(n)$ is any function,  which satisfies  $1/f(n) = o(1)$ and $D$ is some constant.
\end{theorem}

\begin{proof}
    There are two consecutive reductions:
    \begin{enumerate}
    \item By Corollary \ref{colr}, the $\CNN{c}{d}$ can be reduced to $d$ instances of the $\CNN{\alpha_1}{ k_1}$.
        \item By Corollary \ref{colr2}, the $\CNN{\alpha_1}{k_1}$ can be reduced to $k_1^L$ instances of the \newline \mlp{}$(\alpha_2, k_2, L)$.
    \end{enumerate}

    Accordingly, we set:
    \begin{enumerate}
        \item $\alpha_1 = c/2$ and $k_1 = \lceil D_1\log{n} \rceil$ in the first reduction

        \item $\alpha_2 = c/4$, $k_2 = \lceil D_2 L^{-1} \log{n} \rceil \le \lceil D_2 f(n)
            \log{\log{n}} \rceil$ and $L = \lceil \frac{\log{n}}{f(n)\log{\log{n}}}
            \rceil$ in the second reduction.

    \end{enumerate}
    The constants $D_1$ and $D_2$ are chosen to satisfy Corollaries \ref{colr} and \ref{colr2}.
    $k_1^L$ can be bounded in the following way:
     $$k_1^L = \lceil D_1\log{n} \rceil ^{L} = \tilde \Oc(n^{1/f(n)}) =  \tilde \Oc(n^{o(1)}).$$
    The final query complexity, up to factors logarithmic in $n$, equals:

    $$ \que{c}{d} =  \tilde{\Oc}(d^2 + d\ \que{\alpha_1}{ k_1}) =$$
    $$ \tilde{\Oc}(d^2 + d \big( L k_1^2 + k_1^L\ query_{\text{\ml{}}}(\alpha_2, L, k_2)  \big)) = $$
    $$ \tilde{\Oc}(d^2 + d n^{1/f(n)}\ query_{\text{\ml{}}}(\alpha_2, L, k_2)) =  $$
    $$ \tilde{\Oc}(d^2 + d n^{1/f(n)} (k_2L|P | + k_2 \log{n} + k_2L) =  $$
    $$ \tilde{\Oc}(d^2 + d n^{1/f(n)} |P|)=$$
    $$ \tilde{\Oc}(d^2 + d n^{o(1)} |P|).$$

    \noindent The final pre-processing time equals:
    $$\pre{c}{d} =\tilde{\Oc}(d^2 n  + d\ \pre{\alpha_1}{k_1})=$$
    $$\tilde{\Oc}(d^2 n  + d\ \big( L k_1^2 n  + k_1^L\ preproc_{\text{\ml{}}}(\alpha_2, L, k_2)\big)=$$
    $$\tilde{\Oc}(d^2 n  + dn^{1/f(n)}\ preproc_{\text{\ml{}}}(\alpha_2, L, k_2)=$$
    $$\tilde{\Oc}(d^2 n  + d n^{1 + \frac{\ln(3)}{\ln(c/\kappa)} + 1/f(n)}), $$
    where $\kappa = 2\sqrt{k_2} = D\sqrt{f(n)\log{\log{n}}}= \mu$.

\end{proof}
The function $f(n)$ may be chosen arbitrarily. Slowly increasing $f(n)$ will be
chosen for small $c$ close to $\Theta(\log{\log{n}})$. For larger $c$, one
should choose the maximal possible $f(n)$, to optimize the query time complexity.

\section{Conclusion and Future Work}

We have presented the $c$-approximate nearest neighbor algorithm without false negatives in $l_2$ for  any $c = \omega(\sqrt{\log{\log{n}}})$. 
Such an algorithm might work very well for high entropy datasets, where the distances tend to be relatively large (see \cite{Pagh15} for more details). 
Also, we showed that the $c$-approximate nearest neighbor search in $l_2^d$ may be reduced
to $d$ instances of the problem in $l_2^{\log n}$. Hence, further research might focus on
the instances with dimension logarithmic in $n$.

Another open problems are to reduce the time complexity of the algorithm
and relax the restrictions on the approximation factor~$c$
or proving that these
restrictions are essential. We wish to match the time complexities given in
\cite{motwani} or show that the achieved bounds are optimal.

\section{Acknowledgments}
We would like to  thank  Andrzej Pacuk, Adam Witkowski and Kamila Wygocka for meaningful discussion. 
Also, I would like to thank the anonymous reviewers for many comments which greatly increased the quality of this work. 
This work was supported by grant NCN2014/13/B/ST6/00770 of Polish National Science Center.

\bibliography{bib}

\begin{thebibliography}{10}

\bibitem{DBLP:journals/corr/Ahle17}
Thomas~Dybdahl Ahle.
\newblock Optimal las vegas locality sensitive data structures.
\newblock {\em CoRR}, abs/1704.02054, 2017.
\newblock URL: \url{http://arxiv.org/abs/1704.02054}.

\bibitem{DBLP:journals/cacm/AndoniI08}
Alexandr Andoni and Piotr Indyk.
\newblock Near-optimal hashing algorithms for approximate nearest neighbor in
  high dimensions.
\newblock {\em Commun. {ACM}}, 51(1):117--122, 2008.
\newblock URL: \url{http://doi.acm.org/10.1145/1327452.1327494}, \href
  {http://dx.doi.org/10.1145/1327452.1327494}
  {\path{doi:10.1145/1327452.1327494}}.

\bibitem{data-depended-hashing}
Alexandr Andoni and Ilya Razenshteyn.
\newblock Optimal data-dependent hashing for approximate near neighbors.
\newblock In Rocco~A. Servedio and Ronitt Rubinfeld, editors, {\em Proceedings
  of the Forty-Seventh Annual {ACM} on Symposium on Theory of Computing, {STOC}
  2015, Portland, OR, USA, June 14-17, 2015}, pages 793--801. {ACM}, 2015.
\newblock URL: \url{http://doi.acm.org/10.1145/2746539.2746553}, \href
  {http://dx.doi.org/10.1145/2746539.2746553}
  {\path{doi:10.1145/2746539.2746553}}.

\bibitem{CHAZELLE200824}
Bernard Chazelle, Ding Liu, and Avner Magen.
\newblock Approximate range searching in higher dimension.
\newblock {\em Computational Geometry}, 39(1):24 -- 29, 2008.
\newblock URL:
  \url{http://www.sciencedirect.com/science/article/pii/S092577210700065X},
  \href {http://dx.doi.org/http://dx.doi.org/10.1016/j.comgeo.2007.05.008}
  {\path{doi:http://dx.doi.org/10.1016/j.comgeo.2007.05.008}}.

\bibitem{DBLP:conf/focs/Indyk98}
Piotr Indyk.
\newblock On approximate nearest neighbors in non-euclidean spaces.
\newblock In {\em 39th Annual Symposium on Foundations of Computer Science,
  {FOCS} '98, November 8-11, 1998, Palo Alto, California, {USA}}, pages
  148--155, 1998.
\newblock URL: \url{http://dx.doi.org/10.1109/SFCS.1998.743438}, \href
  {http://dx.doi.org/10.1109/SFCS.1998.743438}
  {\path{doi:10.1109/SFCS.1998.743438}}.

\bibitem{Indyk:2007:UPE:1250790.1250881}
Piotr Indyk.
\newblock Uncertainty principles, extractors, and explicit embeddings of l2
  into l1.
\newblock In {\em Proceedings of the Thirty-ninth Annual ACM Symposium on
  Theory of Computing}, STOC '07, pages 615--620, New York, NY, USA, 2007. ACM.
\newblock URL: \url{http://doi.acm.org/10.1145/1250790.1250881}, \href
  {http://dx.doi.org/10.1145/1250790.1250881}
  {\path{doi:10.1145/1250790.1250881}}.

\bibitem{motwani}
Piotr Indyk and Rajeev Motwani.
\newblock Approximate nearest neighbors: Towards removing the curse of
  dimensionality.
\newblock In {\em Proceedings of the Thirtieth Annual ACM Symposium on Theory
  of Computing}, STOC '98, pages 604--613, New York, NY, USA, 1998. ACM.
\newblock URL: \url{http://doi.acm.org/10.1145/276698.276876}, \href
  {http://dx.doi.org/10.1145/276698.276876} {\path{doi:10.1145/276698.276876}}.

\bibitem{johnson84extensionslipschitz}
William Johnson and Joram Lindenstrauss.
\newblock Extensions of {L}ipschitz mappings into a {H}ilbert space.
\newblock In {\em Conference in modern analysis and probability (New Haven,
  Conn., 1982)}, volume~26 of {\em Contemporary Mathematics}, pages 189--206.
  American Mathematical Society, 1984.

\bibitem{O'Donnell:2014:OLB:2600088.2578221}
Ryan O'Donnell, Yi~Wu, and Yuan Zhou.
\newblock Optimal lower bounds for locality-sensitive hashing (except when q is
  tiny).
\newblock {\em ACM Trans. Comput. Theory}, 6(1):5:1--5:13, March 2014.
\newblock URL: \url{http://doi.acm.org/10.1145/2578221}, \href
  {http://dx.doi.org/10.1145/2578221} {\path{doi:10.1145/2578221}}.

\bibitem{pacukatall}
Andrzej Pacuk, Piotr Sankowski, Karol Wegrzycki, and Piotr Wygocki.
\newblock Locality-sensitive hashing without false negatives for l{\_}p.
\newblock In {\em Computing and Combinatorics - 22nd International Conference,
  {COCOON} 2016, Ho Chi Minh City, Vietnam, August 2-4, 2016, Proceedings},
  pages 105--118, 2016.
\newblock URL: \url{http://dx.doi.org/10.1007/978-3-319-42634-1_9}, \href
  {http://dx.doi.org/10.1007/978-3-319-42634-1_9}
  {\path{doi:10.1007/978-3-319-42634-1_9}}.

\bibitem{Pagh15}
Rasmus Pagh.
\newblock Locality-sensitive hashing without false negatives.
\newblock In {\em Proceedings of the Twenty-seventh Annual ACM-SIAM Symposium
  on Discrete Algorithms}, SODA '16, pages 1--9, Philadelphia, PA, USA, 2016.
  Society for Industrial and Applied Mathematics.
\newblock URL: \url{http://dl.acm.org/citation.cfm?id=2884435.2884436}.

\bibitem{Williams2005}
Ryan Williams.
\newblock A new algorithm for optimal 2-constraint satisfaction and its
  implications.
\newblock {\em Theor. Comput. Sci.}, 348(2):357--365, December 2005.
\newblock URL: \url{http://dx.doi.org/10.1016/j.tcs.2005.09.023}, \href
  {http://dx.doi.org/10.1016/j.tcs.2005.09.023}
  {\path{doi:10.1016/j.tcs.2005.09.023}}.

\bibitem{wygos}
Piotr Wygocki.
\newblock On fast bounded locality sensitive hashing.
\newblock {\em ArXiv e-prints}, 2017.
\newblock URL: \url{http://arxiv.org/abs/1704.05902}.

\end{thebibliography}




\end{document}